\documentclass[10pt]{amsart}

\usepackage{amsfonts}
\usepackage{pdfsync}
\usepackage{graphicx}

\usepackage{amsmath}
\usepackage{setspace}
\usepackage[english]{babel}
\usepackage[left=3cm,right=3cm,top=2.5cm,bottom=3cm]{geometry}
\usepackage{amsmath,amssymb,amsthm,bbm,color,graphics,version}
\usepackage{mathrsfs}

\newtheorem{theorem}{Theorem}
\newtheorem{assumption}{Assumption}

\newtheorem{condition}[theorem]{Condition}
\newtheorem{conjecture}[theorem]{Conjecture}
\newtheorem{corollary}[theorem]{Corollary}

\newtheorem{definition}[theorem]{Definition}
\newtheorem{example}[theorem]{Example}

\newtheorem{lemma}[theorem]{Lemma}

\newtheorem{proposition}[theorem]{Proposition}
\newtheorem{remark}[theorem]{Remark}

\newcommand{\bearno}{\begin{eqnarray*}}
\newcommand{\enarno}{\end{eqnarray*}}

\newcommand{\cA}{\mathcal{A}}

\newcommand{\E}{\mathbb{E}}

\newcommand{\R}{\mathbb{R}}
\newcommand{\p}{\mathbb{P}}

\newcommand{\bZ}{Z^{\beta}}

\newcommand{\bR}{\tilde{R}^{\beta}}
\newcommand{\bp}{p_{\beta}}
\newcommand{\bs}{\beta^*}

\newcommand{\bg}{g^{\beta}}
\newcommand{\vbs}{v_{\beta^*}}
\newcommand{\vb}{v_{\beta}}

\newcommand{\tZ}{\tilde{Z}}
\newcommand{\tX}{\tilde{X}}
\newcommand{\tR}{\tilde{R}}
\newcommand{\tp}{\tilde{p}}
\newcommand{\tg}{\tilde{g}}
\newcommand{\ttau}{\tilde{\tau}}

\title{On the Optimal Dividend Problem in the Dual Model with Surplus-Dependent Premiums}

\author{Ewa Marciniak}
\address{AGH University of Science and Technology in Cracow, Faculty of Applied Mathematics, Poland.}
\email{marciniak.ew@gmail.com}

\author{Zbigniew Palmowski}
\address{Mathematical Institute, University of Wroc{\l}aw, Poland.}
\email{zbigniew.palmowski@gmail.com}

\thanks{This work is partially supported by the National Science Centre under the grant 2015/17/B/ST1/01102.
Second author kindly acknowledges partial support by the project RARE -318984, a Marie Curie IRSES Fellowship within the 7th European
Community Framework Programme.}

\date{\today}
\subjclass[2010]{60G51, 60G50, 60K25, 93E20} %
\keywords{}

\begin{document}

\begin{abstract}

This paper concerns the dual risk model, dual to the risk model for insurance applications, where premiums are surplus-dependent. In such a model premiums are regarded as costs, while claims refer to profits.  We calculate the mean of the cumulative discounted dividends paid until ruin, if the barrier strategy is applied. We formulate associated Hamilton-Jacobi-Bellman equation and identify sufficient conditions for a barrier strategy to be optimal.
Some numerical examples are provided when profits have exponential law.

\vspace{3mm}

\noindent {\sc Keywords.} Dividends $\star$ PDMP $\star$ optimal strategy $\star$ barrier strategy $\star$ integro-differential HJB equation $\star$ Dual model $\star$ stochastic control $\star$ Exit problems

\end{abstract}

\maketitle

\pagestyle{myheadings} \markboth{\sc E.\ Marciniak --- Z.\ Palmowski}
{\sc Dual dividend problem}

\vspace{1.8cm}

\tableofcontents

\newpage

%\begin{abstract}

%\end{abstract}

\section{Introduction}\label{sec:intr}
The optimal dividend problem concerns maximizing
the expected discounted cumulative dividend paid up to the time of ruin.
The classical optimal dividend problem has been considered by many authors since De
Finetti \cite{Finetti}, who first introduced the barrier strategy, in which all surpluses above a given level are transferred to a beneficiary,
and raised the question of optimizing its level.
Gerber and Shiu \cite{gerber2004optimal}, Asmussen and Taksar \cite{asmussen1997controlled} and Jeanblanc and Shiryaev \cite{jeanblanc1995optimization} considered the the optimal dividend problem in the Brownian motion setting. Azcue and Muler \cite{azcue2005optimal} and Schmidli \cite{schmidli2006optimisation} studied the optimal dividend strategy under the Cram\'{e}r-Lundberg model using a Hamilton-Jacobi-Bellman (HJB) system of equations. Further Avram et.al. \cite{avram2007optimal, avram2015gerber}, Kyprianou and Palmowski \cite {kyprianou2007distributional}, Loeffen \cite{loeffen2008optimality, loeffen2009optimal}, Loeffen and Renaud \cite{loeffen2010finetti}, Czarna and Palmowski \cite{czarna2010dividend} and many other authors analyze the dividend problem
for the L\'{e}vy risk process using the probabilistic approach.

In this paper, we deal with so-called dual risk process.  In the dual model, in contrast to the classical model, the premium rate is viewed as a cost function and thus it is negative. While claims, on the other hand, should be considered as profits or gains and therefore, they make the surplus increase; see \cite{AvanGerShiu}, \cite{AvanziGerber},  \cite{Yamazaki}, \cite{Ng}.
There are various interpretations for the dual model. For instance, it is appropriate for
the company which pays continuously expenses relevant to research and occasionally gains some random income from some inventions or discoveries.
As an example one can consider pharmaceutical companies, real estate agent offices or brokerage firms that sell mutual funds or insurance products with a front-end load.
For more detailed information, we refer the reader to \cite{AvanGerShiu}.
There is a good deal of work being done on dividend barrier in the dual model under the assumption that the cost function is constant and gains are modeled by a compound Poisson process. Avanzi et al. \cite{AvanGerShiu} considered cases when profits or gains follow an exponential or a mixture of exponential distributions and they derived
explicit formulas for the expected discounted dividend values. They also found the optimal barrier level. The connection between dual and classical model was explained and used by Afonso et al. \cite{Afonso}. Avanzi and Gerber \cite{AvanziGerber} studied a dual model perturbed
by a diffusion. Using Laplace transform method  they determined the optimal strategy among the set of all barrier strategies.
Bayraktar et al. \cite{Yamazaki} used the fluctuation theory to prove the optimality of a barrier strategies for all spectrally positive L\'{e}vy processes.
They characterized the optimal barrier using a so-called scale functions. Moreover, they identify the solution to the dividend problem with capital injections.
Finally Albrecher et al. \cite{Albrecher} using Laplace transforms examined a dual model in the presence of tax payments.

In this paper, we will analyze the dividend problem in a dual model with a
reserve-dependent risk process. We will use the theory of piecewise deterministic Markov processes (PDMP) and martingale properties.
Assuming the absence of transaction costs, we find the corresponding HJB system.
In the next step, we find the optimal barrier strategy and we will
give sufficient conditions for the barrier strategy to be optimal.
The corresponding classical model has been already analyzed in \cite{MarPal}.
In this paper, we also show connections between both models.
As a side product we obtained some exit problems formulas, which can be used to solve problems with capital injections (see \cite{Yamazaki}).

The paper is organized as follows.
In Section \ref{sec:dual}, we introduce the basic
  notation and we describe the dual model we deal with. Section \ref{sec:classical} is dedicated to the corresponding classical model, where we show results on the exit and capital injections problems.
In Section \ref{sec:main}, we present the
  Verification Theorem. We also analyze the barrier strategy and give sufficient conditions for the barrier strategy to be
  optimal among all admissible strategies.
In Section \ref{sec:proofs}, we give all the proofs.
Section \ref{sec:examples} is devoted to some examples.

\section{The Dual Model}\label{sec:dual}

In the dual model the surplus process $R$  (without payment of dividends) with an initial capital $x>0$, solves the following differential equation:
\begin{equation*}
R_t=x-\int_0^t p(R_s)ds+\sum_{k=1}^{N(t)}C_k,
\end{equation*}
where $p$ is a given deterministic, absolutely continuous, positive cost function and
$\left\{C_k\right\}_{i=1}^\infty$ is a sequence of i.i.d. positive random variables with distribution function $F$,
representing the capital injections. Above $N$
is an independent Poisson process with intensity $\lambda>0$ modeling the times at which the capital injections occur.
We assume that $R_t\rightarrow\infty$ a.s. as $t\rightarrow\infty$, $F$ is absolutely continuous with respect of Lebesgue measure with density $f$ and $\E\,C_1<\infty$.

To approach the dividend problem, we consider the controlled surplus process $X^\pi$ satisfying the following stochastic differential equation:
\begin{equation*}
X_t^\pi=x-\int_0^t p(X^\pi_s)ds+\sum_{k=1}^{N(t)}C_k-L_t^\pi,
\end{equation*}
where $\pi$ is a strategy from the class $\Pi$ of all ,,admisibble'' dividend controlls resulting in the cummulative amount of dividends $L_t^\pi$ paid up to time $t$. We say that a dividend strategy $\pi$ is admissible, if the process $L^\pi$ is c\`{a}dl\`{a}g, nondecreasing, starting from zero ($L_{0-}^\pi=0$) and adapted to the natural filtration of the risk process $R$ that satisfies the usual conditions. Moreover, at any time $t$ the dividend payment is smaller than the size of the available reserves ($L^\pi_{t}-L^\pi_{t-}\leq X^\pi_{t-}$).

The object of interest is the averange value of the cumulative discounted dividends paid up to the ruin time:
\begin{equation}\label{vpi}
v_\pi(x):=\E_x\left[\int_0^{\sigma^\pi}e^{-qt}dL^\pi_t\right],
\end{equation}
where $\sigma^\pi:=\inf\{t\geq 0: X^\pi_t\leq 0\}$ is the ruin time and $q\geq 0$ is a given discount rate.
We adopt the convention that $\E_x$ is the expectation with respect to $\p_x(\cdot):=\p(\cdot|X^\pi_{0-}=x)$. Note that unless otherwise stated we write $\sigma$ instead of $\sigma^\pi$ to simplify the notation.

The dividend optimization problem consists in finding the so-called value function $v$ given by
\begin{equation}\label{dualpr}
v(x):=\sup_{\pi\in\Pi}v_\pi(x)
\end{equation}
and the optimal strategy $\pi^*\in\Pi$ such that $$v(x)=v_{\pi^*}(x)\quad\text{for }x\geq 0.$$

In the next subsection we present some of the results for the classical model that can be used to solve the dual problem.

\section{The Classical model}\label{sec:classical}

The connection between the classical and dual model is crucial, since we will use some methods and results already derived for the classical model.

In the classical model, we consider the surplus process $\tR$ (before any regulation) with an initial capital $\tilde{x}>0$, described by the following stochastic differential equation:
\begin{equation*}
\tR_t=\tilde{x}+\int_0^t \tp(\tR_s)ds-\sum_{k=1}^{N(t)}C_k,
\end{equation*}
 where $\tp$ is a given deterministic, positive and absolutely continuous premium function.
Here, the generic $C$ denotes the claim size or lost that arrives according to Poisson process $N$.
Let $\tX_t$ be a reflected process satisfing the equation
\begin{equation*}
\tX_t=\tilde{x}+\int_0^t \tp(\tX_s)ds-\sum_{k=1}^{N(t)}C_k+\tilde{L}_t^{0},
\end{equation*}
where $\Delta \tilde{L}_t^{0}=\left(-\tX_t\right)\mathbb{I}_{\{\tX_t<0\}}$. Note that $\tilde{L}^{0}$ is a nondecreasing adapted process with $\tilde{L}_{0-}^0=0$. Above $\tilde{L}_t^0$ represents the total amount of capital injections up to time $t$.

We will first demonstrate how to identify
\begin{equation*}
\E_{x}\left[\int_0^{\tilde{T}_a^+}e^{-qt}d\tilde{L}^{0}_t\right],
\end{equation*}
where $\tilde{T}_a^+:=\inf\{t\geq 0:\tX_t\geq a\}$.
It is the average discounted value of injected capital paid up to the time of reaching the barrier $a>0$ by the controlled process $\tX$.

Two functions, $\tilde{W}_q$ and $\tilde{G}_{q,w}$, are crucial to solve this problem. There are related with two-sided and one-sided exit problems of $\tR$ in the following way:
\begin{equation}\label{exit1}
\E_x\left[e^{-q\ttau_a^+}\mathbb{I}_{\{\ttau_a^+<\ttau^-_0\}}\right]=\frac{\tilde{W}_q(x)}{\tilde{W}_q(a)},\qquad\text{for }x\in[0,a]
\end{equation}
\begin{equation}\label{exit2}
\tilde{G}_{q,w}(x):=\E_x\left[e^{-q\ttau_{0}^-}w(|\tR_{\ttau^-_{0}}|)\mathbb{I}_{\{\ttau_{0}^-<\infty\}}\right],\qquad\text{for }x>0,
\end{equation}
where $a>0$, $\ttau_a^+=\inf\{t\geq 0: \tR_t\geq a\}$ and $\ttau_a^-=\inf\{t\geq 0: \tR_t< a\}$. Function $\tilde{G}_{q,w}$ is defined for some general positive penalty function $w$.
For the existence and properties of the functions $\tilde{W}_q$ and $\tilde{G}_{q,w}$ we refer the reader to \cite{Corina} and \cite{MarPal}.

Note that $\tR$ is a piecewise deterministic Markov process.
By $\tilde{\cA}$ we denote the full generator of $\tR$, i.e. we have:
$$\tilde{\cA} m(x)=\tp(x)m'(x)+\lambda\int_0^\infty (m(x-z)-m(x))dF(z)$$
acting on absolutely continuous functions $m$ such that
\begin{equation*}\label{domain}
\E\bigg[\sum_{\sigma_i\leq t}|m(\tR_{\sigma_i})-m(\tR_{\sigma_i-})|\bigg]<\infty
\qquad\text{for any}\ t\geq 0.
\end{equation*}
Above $\{\sigma_i\}_{i\in \mathbb{N}\cup\{0\}}$ denotes the times at which the claims occur (see Davis \cite{Davis} and Rolski et al. \cite{Rolski}). Moreover, $m'$ denotes the  density of~$m$.
Note that any function, which is absolutely continuous and ultimately dominated by an affine function, is in the domain of the full generator $\cA$ as a consequence
of the assumption that $\E C_1<\infty$.
Recall that for any function $m$ from the domain of  $\cA$ the process
$$\bigg\{e^{-qt}m(\tR_t) -\int_0^te^{-qs}\left(\tilde{\cA}-q\mathbf{I}\right) m(\tR_s)\,ds, t\geq 0\bigg\}$$
is a martingale.

Marciniak and Palmowski \cite{MarPal} showed that, if the claim size distribution of $C$
is absolutely continuous, then functions $\tilde{W}_q$ and $\tilde{G}_{q,w}$ are differentiable and satisfy equations:
\begin{equation}\label{eq:h}
\tilde{\cA} \tilde{W}_q(x)=q\tilde{W}_q (x)\quad \text{for } x\geq 0,\qquad  \tilde{W}_q (x)=0\quad \text{for } x<0, \quad \tilde{W}_q(0)=1
\end{equation}
and
\begin{equation}\label{eq:G}
\tilde{\cA} \tilde{G}_{q,w}(x)=q\tilde{G}_{q,w} (x) \quad \text{for }x\geq 0,\qquad  \tilde{G}_{q,w} (x)=w(x)\quad \text{for } x<0.
\end{equation}
Now we are drawing our attention to the exit problem for the reflected process $\tilde{X}$. The first passage time of a positive level $a>0$ we denote by
$$T_a^+:=\inf\{t\geq0:\tX_t> a\}.$$
The following result express the Laplace transform of the exit time $T^+_a$ in terms of the functions $\tilde{W}_q$ and $\tilde{G}_{q,1}$.

\begin{theorem}\label{e-qT_a}
Let $a>0$, $x\in[0,a]$ and $q\geq 0$. If $T^+_a<\infty$ $\p$-a.s. then
$$\E_x\left[e^{-qT_a^+}\right]=\tZ(x)/\tZ(a),$$ where
\begin{equation}\label{Z}
\tZ(x) :=\big(1-\tilde{G}_{q,1}(0)\big)\tilde{W}_q(x)+\tilde{G}_{q,1}(x)\qquad\text{for }x\geq 0
\end{equation}
and $\tZ(x)=1$ for $x<0$.
\end{theorem}

\begin{proof}
By the properties of $ \tilde{G}_{q,1}$ and $\tilde{W}_q$
\begin{equation}\label{eqZ}
(\tilde{\cA}-q\mathbf{I})\tZ(x)=0\quad \text{for all } x> 0,\qquad \tZ(x)=1 \text{ for } x\leq 0.
\end{equation}
Note that $\tZ$ is continuous on $\R$ and continuously differentiable on $\R\setminus \{0\}$ with right and left derivative at $0$. Thus $\tZ$ is absolutly continuous on $\R$. Moreover, $\tZ$ is dominated by some affine function on the set $(-\infty,a)$. Therefore, $\tZ$ is in the domain of the
piecewise deterministic Markov process $\tX_{t\wedge T^+_a}$ (see Davis \cite{Davis} and Rolski et al. \cite{Rolski}).
This means that the process
\begin{eqnarray*}
K_{t\wedge T_a^+}:=e^{-q(t\wedge T^+_a)}\tZ(\tX_{t\wedge T^+_a})-\int_0^{t\wedge T^+_a} e^{-qs}\Big(\tilde{\cA}^X-q\mathbf{I}\Big)\tZ(\tX_{s})\,ds
\end{eqnarray*}
is a martingale for the generator $\tilde{\cA}^X$ of $\tX_{t\wedge T^+_a}$ given by:
$$\tilde{\cA}^Xm(x)=\tilde{p}(x)m'(x)+\lambda\int_0^x(m(x-z)-m(x))dF(z)+\lambda(m(0)-m(x))\p(C_1>x).$$

Note that $\tZ(x-z)=\tZ(0)$ for all $z\geq x$. Thus, by (\ref{eqZ}), we have $K_{t\wedge T_a^+}=e^{-q(t\wedge T^+_a)}\tZ(\tX_{t\wedge T^+_a})$, which
is bounded by $\sup_{x\in[0,a]}\{\tZ(x)\}<\infty$. Hence $K_{t\wedge T_a^+}$ is a uniformly integrable (UI) martingale and
$$\tZ(x)=\E_x[K_{t\wedge T_a^+}]\rightarrow\E_x[K_{T_a^+}]=\tZ(a)\E_x[e^{-qT_a^+}]\text{ as }t\rightarrow\infty.$$
We used in above the assumption that $T_a^+<\infty$ a.s. This completes the proof.
\end{proof}

\begin{remark}\rm
Note that the function $\E_x[e^{-qT_a^+}]$ is increasing on $(0,a)$ for any $a>0$, thus $\tZ(x)$ has to be increasing on $(0,\infty)$.
Moreover, $\tilde{Z}$ is continuous on $\R$ and continuously differentiable on $\R\setminus\{0\}$.
\end{remark}

From \eqref{exit2} and Theorem \ref{e-qT_a}, using the
strong Markov property of $\tR$ and $\tX$, one can
obtain the following crucial identities.

\begin{theorem}\label{tg}
(i) For $x\geq 0$ and $q>0$ it holds that
\begin{eqnarray}
\tg(x)&:=&\E_x\left[\int_0^\infty e^{-qs}d\tilde{L}^{0}_s\right]=-\E_x\left[e^{-q\ttau_{0}^-}\tR_{\ttau_{0}^-}\right]+\E_x\left[e^{-q\ttau_{0}^-}\right]\E_{0}\left[\int_0^\infty e^{-qs}d\tilde{L}^{0}_s\right]\nonumber\\
&&= \tilde{G}_{q,|x|}(x)+\tilde{G}_{q,1}(x)\tg(0).\label{gmale}
\end{eqnarray}
(ii) Let $a\in(0,\infty)$. For all $x\in[0,a)$ we have
\begin{equation}\label{barrier}
\E_x\left[\int_0^{T_a^+}e^{-qs}d\tilde{L}^{0}_s\right]=\tg(x)-\tZ(x)g(a)/\tZ(a).
\end{equation}
\end{theorem}

\section{Optimal Dividend Strategy}\label{sec:main}

\subsection{HJB equations}
To prove the optimality of a particular dividend strategy $\pi$ among all admissible strategies $\Pi$ for the dual problem (\ref{dualpr}), we consider the following Hamilton-Jacobi-Bellman (HJB) system:
\begin{eqnarray}\label{hjb}
\begin{split}
\max \left\{(\cA-q\mathbf{I})m(x),1-m'(x))\right\}&=&0\quad\text{for }x>0,\\
 m(x)&=&0 \quad\text{for } x\leq 0,
 \end{split}
\end{eqnarray}
where $\cA$ is the full generator of the piecewise deterministic Markov process $R$ given by:
\begin{equation}\label{cAdual}
\cA m(x)=-p(x)m'(x)+\lambda\int_0^\infty(m(x+y)-m(x))\, dF(y),
\end{equation}
acting on absolutely continuous functions $m$ such that
\begin{equation*}
\E\bigg[\sum_{\sigma_i\leq t}|m(R_{\sigma_i})-m(R_{\sigma_i-})|\bigg]<\infty
\qquad\text{for any}\ t\geq 0,
\end{equation*}
where $\{\sigma_i\}_{i\in \mathbb{N}\cup\{0\}}$ denotes the times at which the claims occur (see Davis \cite{Davis} and Rolski et al. \cite{Rolski}). In this case $m'$ denotes the  density of~$m$.
Recall that any function, which is absolutely continuous and ultimately dominated by an affine function, is in the domain of the full generator $\cA$.

\begin{theorem}\label{verif} \textit{(Verification Theorem)}
Assume that $m:[0,\infty)\rightarrow[0,\infty)$ is continuous and $m(0)=0$. Extend $m$ to the negative half-axis by $m(x)=0$ for $x<0$. Suposse that $m$ is $C^1$ on $(0,\infty)$. If $m$ satisfies (\ref{hjb}), then $m\geq v$ on $(0,\infty)$.
\end{theorem}

The proof of Theorem \ref{verif} is given in Appendix.

\subsection{Barrier Strategies}

In this subsection, we consider the barrier strategy $\pi_\beta$,
which pays everything above a fixed level $\beta\geq 0$ as dividends.
We find $v_{\pi_\beta}$ defined by (\ref{vpi}).
To this end, we will use methods from the classical model.
\begin{figure}
\begin{center}

\includegraphics[scale=0.7]{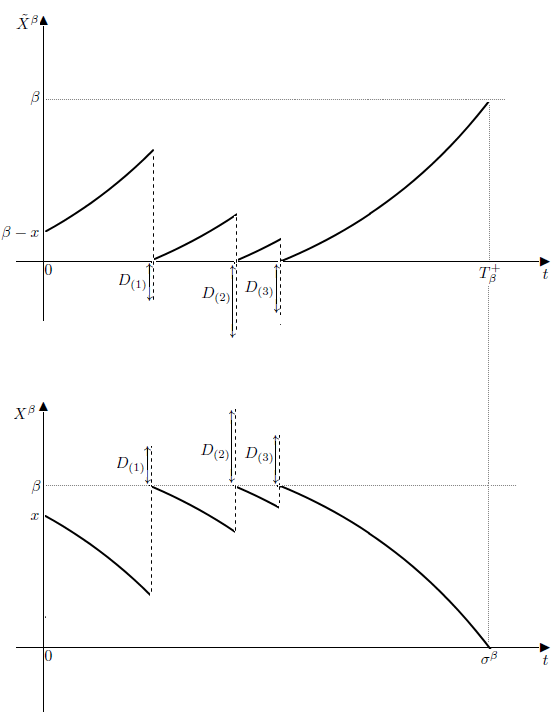}
\end{center}

\caption{Classical vs. dual model. }
\label{figure:fig1}
\vspace{-10pt}
\end{figure}

Consider the classical risk process $\bR$
with
\[\tp(\cdot)=\bp(\cdot):=p(\beta-\cdot),\qquad \tilde{x}=\beta-x.\]
The idea is to transfer the jumps over barrier $\beta$
in the dual model into the injections in the classical model
that happen when it gets below zero.
The Figure~\ref{figure:fig1} shows this connection between the classical and the dual model.
Using this correspondence, by a direct application of the results given in \eqref{barrier}, we have the following
\ of the value function under the barrier strategy.
\begin{theorem}\label{valuebarrier}
We have,
\begin{eqnarray}\label{va}
v_{\pi_\beta}(x):=v_\beta(x) &=&
\begin{cases} \bg(\beta-x)-\bZ(\beta-x)\bg(\beta)/\bZ(\beta)
& 0\leq x \leq \beta,\\
&\\
x-\beta+v_\beta(\beta), & x > \beta,
\end{cases}
\end{eqnarray}
where $\bg (\beta-x):=\tg(\tilde{x})$ is given in \eqref{gmale}.
\end{theorem}

\begin{remark}\label{newremark}\rm
Note that $v_\beta$ is continuous on $[0,\infty)$ with $v_\beta(0) = 0$.
Note also that $\vb$ solves an integro-differential equation:
\begin{equation}\label{eqvb}
(\cA-q\mathbf{I})\vb(x)=0\qquad\text{for }x\in(0,\beta).
\end{equation}
Moreover, $v_\beta$ is $C^1$ on $(0,\beta)$ and $v_\beta'$, $v_\beta$ are left-contiuous at $\beta$ and right-continuous at $0$.
\end{remark}

We start from the case when optimal barrier is located at $0$, that
is, according to this strategy it is optimal to pay all initial capital as dividends.

\begin{theorem}
If $-p(x)+\lambda\E C_1-qx\leq 0$ for all $x\geq 0$ then $v_0(x)=v(x)$ for all $x$ and the barrier strategy $\pi_\beta$ with $\beta=0$ is optimal.
\end{theorem}

\proof
Consider $h(x)=x$. We have $$(\cA-q)h(x)=-p(x)+\lambda\int_0^\infty zdF(z)-qx\leq 0.$$ Thus $h$ satisfies (\ref{hjb}), and by Theorem \ref{verif} we have $h\geq v$ on $(0,\infty)$. However, $h$ is the value function for the barrier strategy $\pi_0$. Thus  $h(x)=v(x)$ for all $x\geq 0$ and $\pi_0$ is indeed optimal.
\qed

\noindent

We shall examine the smoothness of $v_\beta$ at $x=\beta$ in order to find a candidate $\beta^*$ for the optimal barrier level.
In particular, we will show that the optimal barrier $\bs$ is given by:
\begin{equation}\label{betastar}
\bs:=\inf\left\{\beta\geq 0:\vb'(\beta-)=1\right\}.
\end{equation}

The next theorem concerns the existence of $\bs$.
\begin{theorem}\label{bs exist}
Assume that $p$ is a $C^1$ function on $[0,\infty)$ such that $\lambda\E C_1>p(0)>0$ and there exists $\hat{x}>0$ such that $p(\hat{x})>\lambda\E C_1-q\hat{x}$.

Then $\bs$ exists and $\bs\in(0,\hat{x})$.
\end{theorem}

In next result we give sufficient conditions for the barrier strategy to be optimal.
\begin{theorem} \label{T:optimal} Assume that $\bs<\infty$ exists and $p\in C^1$. Let $-p'(x)-q<0$ for all $x\leq\bs$.
Then the barrier strategy $\pi_{\beta^*}$ is optimal and $v(x)=\vbs(x)$ for all $x\geq 0$.
Moreover, in this case $v(x)$ uniquely solves the equation
\eqref{eqvb} with the boundary condition $v^\prime(\beta^*)=1$.
\end{theorem}
\begin{remark}\rm
Note that in the case when $p(x)=p$ is constant all assumptions are satisfied and the barrier strategy is always optimal.
It was already proved in \cite{Yamazaki}.
In the case of general premium function we conjecture that this is not always true (even if
$\beta^*$ is well-defined). Unfortunately, due to complexity of the HJB  equation in this case, it's difficult to construct a counterexample.
\end{remark}

\section{Examples}\label{sec:examples}
In this section, we assume that injection size $C_1$ has an exponential distribution
with a parameter $\mu>0$. Then the equation (\ref{eqvb}) can be transformed into
\begin{equation}\label{eq:vb}
-p(x)\vb''(x)+(\mu p(x)-p'(x)-\lambda-q)\vb'(x)+\mu q \vb(x)=0
\end{equation}
with the initial conditions $\vb(0)=0$ and
$$p(0)\vb'(0)=\lambda\mu\int_0^\beta \vb(z)e^{-\mu z}\,dz+\mu \int_\beta^\infty(z-\beta+\vb(\beta))e^{-\mu z}\,dz.$$

\textit{Example 1.} Consider an increasing rational cost function given by $$p_1(x)=c\big(2-(1+x)^{-1}\big).$$
This premium function tends to constant $2c$ for large present reserves
and it is within range $[c,2c]$.
Solving equation (\ref{eq:vb}) numerically, we can identify $\beta^*$ and add later some observations.
For calculations we used representation \eqref{va} of the equation for the barrier value function. We skip details here.
Note also that for $c>\lambda/\mu$ from Theorem \ref{bs exist} it follows that the barrier level $\bs$ is well-defined
and by Theorem \ref{T:optimal} the barrier strategy $\pi_{\beta^*}$ is optimal.
In the Tables 1, 2 and 3 we present numerical results.

\begin{table}[h]

\caption{ \small Dependence $\bs$ of $c$.}
\begin{center}
\begin{tabular}{c|c|c|c|c|c|c|c|c}
\hline
\multicolumn{9}{c}{$q=0.1$, $\mu=0.01$, $\lambda=0.1$}\\ \hline
$c$ & 1 & 1.4& 1.6& 2& 2.6& 3    & 4   &4.5   \\ \hline
$\bs$ & 26.5 &32.2& 34.25& 37.1 &38.3& 37.1&26.6&  17.5 \\ \hline
\end{tabular}
\end{center}

\bigskip
\caption{ \small Dependence $\bs$ of $q$.}
\begin{center}
\begin{tabular}{c|c|c|c|c|c|c|c}
\hline
\multicolumn{8}{c}{$\mu=0.01$, $\lambda=0.1$, $c=2$}\\ \hline
$q$ & 0.08& 0.09&0.1 & 0.12 &0.14&0.15&0.17\\ \hline
$\bs$ &49.4 & 42.3 &37.1 & 29.8 & 25.5&23.2&20.25\\ \hline
\end{tabular}
\end{center}

\bigskip
\caption{ \small Dependence $\bs$ of $\mu$.}
\begin{center}
\begin{tabular}{c|c|c|c|c|c|c|c|c}
\hline
\multicolumn{9}{c}{ $q=0.1$,$\lambda=0.1$, $c=2$}\\ \hline
$\mu$   &0.005 & 0.007 & 0.01 & 0.015 & 0.017&0.02& 0.025&0.27\\ \hline
$\bs$ &54.7& 46.6 & 37.1 &24.2 &19.7& 13.1 &4.65&2.93 \\ \hline
\end{tabular}
\medskip
\end{center}

\end{table}

\begin{figure}
\begin{center}

\includegraphics[scale=0.8]{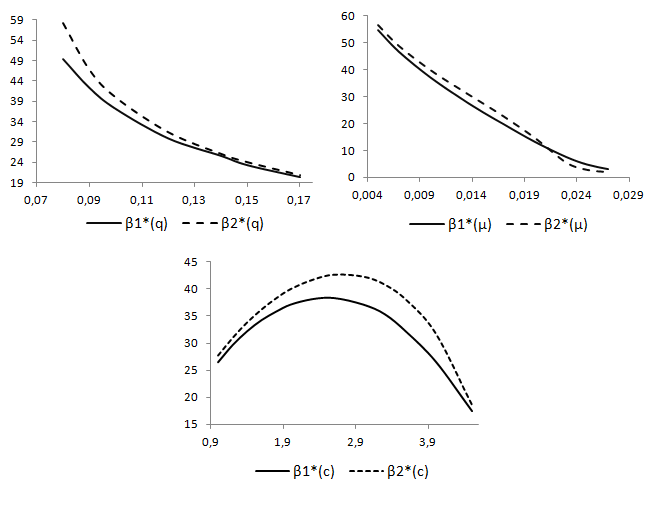}
\end{center}
\vspace{-20pt}
\caption{  The optimal barrier as a function of $q$, $\mu$ and $c$ respectively with $\bs 1$ corresponding to the model with $p_1$ premium function and $\bs 2$ to the model with $p_2$.  }
\label{figure:fig2}
\vspace{20pt}

\end{figure}

\textit{Example 2.}
In this example we investigate the exponential-like cost function $$p_2(x)=2c(1+e^{-x})^{-1}$$
which produces quicker convergence to the constant than in the previous example. Still, we expect similiar properties as both functions are increasing with the same range of values between $c$ and $2c$.
Again the equation (\ref{eq:vb}) produces numerical values of barrier level $\beta^*$ which
is by Theorems \ref{bs exist} and \ref{T:optimal} well-defined and optimal within all admissible strategies.
Figure~\ref{figure:fig2} compare $\bs$ in both examples and as predicted respective values are close.

From above numerical analysis it follows that there is a huge impact on $\beta^*$ of the parameters $q$, $\mu$ and $c$.
Moreover, the optimal barrier level $\beta^*$ decreases rapidly with $\mu$ increasing as well as with increasing $q$ in both cases.
We can also observe concavity of the optimal barrier level with respect to $c$.

\textit{Example 3.}
Let's consider a decreasing premium function given by $$p_3(x)=c+\frac{0.1}{1+x}.$$ Theorem \ref{T:optimal} implies that with $q>0.1$ the barrier strategy $\bs$ is optimal among all admissible strategies. It seems that when it comes to properties of  $\bs$, they remaind the same as in the previous examples (see Figure~\ref{figure:fig3}).

\begin{figure}
\begin{center}
\bigskip
\includegraphics[scale=0.8]{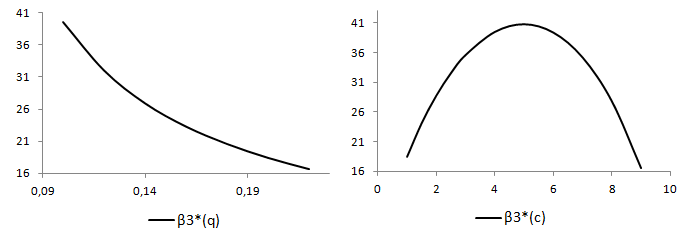}
\end{center}
\caption{  The optimal barrier as a function of $q$ and $c$ respectively with $\bs 3$ corresponding to the premium function $p_3$.}
\label{figure:fig3}
\vspace{-10pt}
\end{figure}
\newpage

\section{Appendix}\label{sec:proofs}

\subsection{Proof of the Verification Theorem \ref{verif}}
 Fix an arbitrary $x\in(0,\infty)$ and $\pi\in\Pi$. By $\{t_n\}_{n=1}^{\infty}$ denote all jumping times for $L^{\pi}$.  Since $m$ is $C^1(0,\infty)$ and $X^{\pi}_{t\wedge \sigma^\pi}\in[0,\infty)$, we are  allowed to use It\^{o}'s formula to the process $Y_t:=e^{-q(t\wedge \sigma^\pi)}m(X^\pi_{t\wedge\sigma^\pi})$, which gives:
\begin{eqnarray*}
Y_t-Y_0&=&\int_0^{t\wedge\sigma^\pi}e^{-qs}(\cA -q\mathbf{I})m(X^\pi_{s-})ds+M_t-\int_0^{t\wedge \sigma^\pi}e^{-qs}dL^{\pi,c}_s\\
&&+\sum_{0\leq t_n\leq t\wedge \sigma^\pi}e^{-qt_n}[m(X^\pi_{t_n-}+C_{N_{t_n}}\Delta N_{t_n}-\Delta L^\pi_{t_n})-m(X^\pi_{t_n-}+C_{N_{t_n}}\Delta N_{t_n})],
\end{eqnarray*}
where $L^{\pi,c}$ denotes the continuous part of $L^{\pi}$ and $M_t$ is a martingale with $M_0=0$.
Since $m$ satisfies (\ref{hjb}) we have
\begin{eqnarray*}
Y_t-Y_0&\leq& \int_0^{t\wedge\sigma^\pi}e^{-qs}(\cA -q\mathbf{I})m(X^\pi_{s-})ds+M_t-\int_0^{t\wedge \sigma^\pi}e^{-qs}dL^{\pi,c}_s-\sum_{0\leq t_n\leq t\wedge \sigma^\pi}e^{-qt_n}\Delta L^\pi_{t_n}\\
&\leq & -\int_0^{t\wedge \sigma^\pi}e^{-qs}dL^{\pi}_s+M_t.
\end{eqnarray*}
Taking expectations and using the fact that $m\geq 0$, we obtain
\begin{eqnarray*}
m(x)&\geq & \E_x \big[e^{-q(t\wedge \sigma^\pi)}m(X^\pi_{t\wedge\sigma^\pi})\big]+\E_x\int_0^{t\wedge \sigma^\pi}e^{-qs}dL^\pi_{s}\\
&\geq &\E_x\int_0^{t\wedge \sigma^\pi}e^{-qs}dL^\pi_{s}.
\end{eqnarray*}
Letting $t\rightarrow\infty$ and applying the Monotone Covergence Theorem gives:
\begin{eqnarray*}
m(x)&\geq &v_\pi(x).
\end{eqnarray*}
Therefore, since $\pi\in\Pi$ and $x$ were arbitrary, we proved the desired inequality \linebreak $m(x)\geq v(x)$ for all $x\in[0,\infty)$.
\qed

\noindent

\subsection{Additional facts}
Before we prove the main results of this paper we demonstrate
few auxiliary facts used in these proofs.

\begin{lemma}\label{hbeta}
Let $\beta\geq 0$. Assume that we have a function $h_\beta:\R\rightarrow [0,\infty)$, such that $h_\beta(x)=x-\beta+h_\beta(\beta)$ for all $x>\beta$ and $h_\beta(x)=0$ for all $x\leq 0$. If the function $h_\beta(x)$ solves the equation $$(\cA-q\mathbf{I})h_\beta(x)=0\qquad\text{ for } 0<x<\beta$$ with the boundary conditions $h_\beta'(\beta)=1$, then $$h_\beta(x)=v_\beta(x)\qquad\text{ for all }x\geq 0.$$
\end{lemma}

 \textit{Proof of Lemma \ref{hbeta}.} Denote $X^\beta:=X^{\pi_\beta}$. Take an arbitrary $x\in[0,\beta]$. Applying Ito's formula for $e^{-qt}h_\beta(X^\beta_t)$, we obtain
\begin{eqnarray*}
\E_x\left[e^{-q(t\wedge\sigma^\beta)}h_\beta(X^\beta_{t\wedge\sigma^\beta})\right]&=&h_\beta(x)+\E_x\left[\int_0^{t\wedge\sigma^\beta} e^{-qs}(\cA-q\mathbf{I})h_\beta(X_s^\beta)ds\right]+\E_x\left[\int_0^{t\wedge\sigma^\beta}e^{-qs}h_\beta(X^\beta_s)dL^{\beta,c}_s\right]\notag\\
&&+\E_x\left[\sum_{0\leq s\leq {t\wedge\sigma^\beta}}e^{-qs}\left(h_\beta(X_{s-}^\beta-\Delta L_s^\beta)-h_\beta(X_{s-}^\beta)\right)\mathbb{I}_{\Delta L_s^\beta>0}\right]\\
&=&h_\beta(x)+\E_x\left[\int_0^{t\wedge\sigma^\beta} e^{-qs}(\cA-q\mathbf{I})h_\beta(X_s^\beta)ds\right]-\E_x\left[\sum_{0\leq s\leq {t\wedge\sigma^\beta}}e^{-qs}\Delta L_s^\beta\mathbb{I}_{\Delta L_s^\beta>0}\right].
\end{eqnarray*}
Note that between the positive jumps the process $X^\beta$ is decreasing and hence $L^{\beta,c}\equiv 0$. Moreover, $X_{s-}^\beta>\beta$ on $\{\Delta L^\beta_s>0\}$ and $h_\beta(X_{s-}^\beta-\Delta L_s^\beta)=h_\beta(\beta)$. Hence, after rearranging and letting $t\rightarrow\infty$, we get
$$h_\beta(x)=\E_x\left[\int_0^{\sigma^\beta}e^{-qs}dL^\beta_s\right]=v_\beta(x).$$
This completes the proof.
\qed

\noindent

\begin{lemma} \label{concave}
Assume that $p$ is $C^1$ on $(0,\infty)$ and that $-p'(x) -q<0$ for all $x\in(0,\beta]$. If $v_\beta'(\beta-)\geq 1$
then $\vb(x)$ is in $C^2$ on $(0,\beta)$ and it is increasing and concave on $(0,\infty)$.
\end{lemma}

\proof
We begin by proving that $\vb$ is increasing on $(0,\beta)$.
Let
$\tau^+_a:=\inf\{t\geq 0:R_t\geq a\}$ and $\tau^-_0:=\inf\{t\geq 0:R_t<0\}$.
By the Strong Markov property of the PDMP $R_t$ for all $0<y<x<\beta$
$$ \vb(y)=\vb(x)\E_y\left[e^{-q\tau_x^+};\tau_x^+<\tau_0^-\right]<v_{\beta(x)},$$
which completes the proof of this statement.
Furthermore, since $p$ is $C^1$ by (\ref{eqvb}), we have that $\vb$ is $C^2$ on $(0,\beta)$
and $\vb$ is left-continuous at $\beta$.
Moreover, if we differentiate (\ref{eqvb}) with respect to x we get:
\begin{equation}
p(x)\vb''(x)=(-p'(x)-\lambda-q)\vb'(x)+\lambda\int_0^{\beta-x}\vb'(x+z)f(z)\, dz+\lambda\int_{\beta-x}^\infty f(z)\, dz.
\end{equation}
Consequently, since $\vb'(\beta-)\geq 1$, we obtain
$$p(\beta)\vb''(\beta-)\leq -p'(\beta)-\lambda-q+\lambda\int_{0}^\infty f(z)\, dz=-p'(\beta)-q<0.$$
Assume now a contrario that $\vb$ is not concave. Then by continuity of $\vb''$, there exists $\hat{x}\in (0,\beta)$,
such that $\vb''(\hat{x})=0$ and $\vb''(x)<0$ for all $x\in (\hat{x},\beta)$. Hence, from the assumption that $v_\beta'(\beta-)\geq 1$,
\begin{eqnarray*}
0=p(\hat{x})\vb''(\hat{x})&=&(-p'(\hat{x})-\lambda-q)\vb'(\hat{x})+\lambda\int_0^{\beta-\hat{x}}\vb'(\hat{x}+z)f(z)\, dz+\lambda\int_{\beta-\hat{x}}^\infty f(z)\, dz\\
&< &(-p'(\hat{x})-\lambda-q)\vb'(\hat{x})+\lambda\vb'(\hat{x})\int_0^{\beta-\hat{x}}f(z)\, dz+\lambda\vb'(\hat{x})\int_{\beta-\hat{x}}^\infty f(z)\, dz\\
&= &(-p'(\hat{x})-q)\vb'(\hat{x})<0.
\end{eqnarray*}
This leads to the contradiction. Therefore, $\vb''(x)<0$ for all $x\in(0,\beta)$.
\qed

We are ready to prove the main results.

\subsection{Proof of the Theorem \ref{bs exist}}
From \eqref{eqvb} and \eqref{cAdual} we can conclude that $v_\beta$ on $(0,\beta)$ satisfies the equation:
\begin{equation}\label{pretrans}
-p(x)v_\beta'(x)-(\lambda+q)v_\beta(x)+\lambda\int_0^{\beta-x}v_\beta(x+z)f(z)\,dz+\lambda\int_{\beta-x}^\infty \left(x+z-\beta+v_\beta(\beta)\right)f(z)\,dz=0
\end{equation}
with the initial condition $v_\beta(0)=0$. We recall that we are looking for $\beta^*$ satisfying $\vbs'(\bs)=1$.
Let \[u_\beta(x):=\vb'(x).\]
Transforming the equation \eqref{pretrans}, we obtain the following Fredholm equation for $u_\beta$:
\begin{equation}\label{ubeta}
u_\beta(x)=G_\beta(x)+\int_0^\beta K(x,y)u_\beta(y)\,dy,\end{equation}
 where
$$G_\beta(x):=\frac{\lambda}{p(x)}\left(\int_{\beta-x}^\infty \left(x+z-\beta\right)f(z)dz\right),$$
\begin{eqnarray*}
K(x,y):=\left\{
\begin{array}{lll}
\frac{-q}{p(x)}&\text{ for }& 0\leq y\leq x\\
\frac{\lambda}{p(x)}\int_{y-x}^\infty f(z)\,dz&\text{ for }&  y>x\geq 0.
\end{array}
\right.
\end{eqnarray*}
Taking $x=\beta$ in \eqref{ubeta} leads to the equation:
\begin{equation}\label{gamma}
u_\beta(\beta)=\frac{\lambda}{p(\beta)}\E(C_1)-\frac{q}{p(\beta)}\int_0^\beta u_\beta(y)\,dy.\end{equation}
We denote:
\[\gamma(\beta):=u_\beta(\beta).\]
We want to prove the existence of $\beta^*$ such that
\begin{equation}\label{stat}
\gamma(\beta^*)=1.\end{equation}
%By Lemma \ref{concave} we have that $u_\beta$ is nonnegative for all $\beta>0$.
By our assumption:
\begin{equation}\label{wiekjeden0}
\gamma(0)=\frac{\lambda}{p(0)}\E(C_1)>1.\end{equation}
We will show that
\begin{equation}\label{wiekjeden}\gamma(\hat{x})\leq 1.\end{equation}
Indeed, assume a contrario that $\gamma(\hat{x})>1$. Then by Lemma \ref{concave}, we have that $u_{\hat{x}}$
is increasing and consequently $u_{\hat{x}}(y)>1$ for all $y\in(0,\hat{x})$. Thus
\begin{eqnarray*}
\gamma(\hat{x})&=&\frac{\lambda}{p(\hat{x})}\E(C_1)-\frac{q}{p(\hat{x})}\int_0^{\hat{x}} u_{\hat{x}}(y)\,dy\\
&\leq& \frac{\lambda}{p(\hat{x})}\E(C_1)-\frac{q}{p(\hat{x})}\hat{x}<1.
\end{eqnarray*}
In this way we derived a contradiction.
To get \eqref{stat} by inequalities \eqref{wiekjeden0} and \eqref{wiekjeden}, it suffices to prove that
$\gamma$ is a continuous function.
The latter follows from the following observation.
We denote $\Delta_\beta b_\beta(x)=b_{\beta}(x)-b_{\beta-}(x)$ for a general function $b_\beta$.
From \eqref{ubeta} it follows that
\[\Delta_\beta u_\beta(\beta)=\int_0^\beta K(\beta,y)\Delta_\beta u_\beta(y)\,dy\]
and hence function $\beta\rightarrow \Delta_\beta u_\beta(\beta)$ is continuous.
Note that from the last conclusion it follows that $\Delta_\beta\Delta_\beta u_\beta(\beta)=\Delta_\beta u_\beta(\beta)=0$ and
function $\beta\rightarrow  u_\beta(\beta)=\gamma(\beta)$ is also continuous.
This completes the proof.
\qed

\subsection{Proof of the Theorem \ref{T:optimal}}
Due to the Lemma \ref{concave}, we have $\vbs'(x)\geq \vbs'(\bs)=1$ for all $x\in(0,\bs)$. Moreover, $(\cA-q\mathbf{I})\vbs(x)=0$ for $x\in(0,\bs)$. It remains to prove that $(\cA-q\mathbf{I})\vbs(x)\leq 0$ for $x>\bs$. Since $\vbs(x)=x-\bs+\vbs(\bs)$ for $x>\bs$, we have:
\begin{equation}\label{eq1}
(\cA-q\mathbf{I})\vbs(x)=-p(x)+\lambda\int_0^\infty zdF(z)-q(x-\bs+\vbs(\bs)).
\end{equation}
Note that $\vbs$ is $C^1$ and therefore, by assumption that $p\in C^1$,
the function $(\cA-q\mathbf{I})\vbs(x)$ is continuous at $x=\bs$. Thus we have $(\cA-q\mathbf{I})\vbs(\bs)=0$. The assumption $-p'(x)-q\leq 0$, together with (\ref{eq1}), give the required inequality.
Hence by Theorem \ref{verif}, $\vbs(x)\geq v(x)$ for all $x\geq0$. At the same time from the definition of the value function
we have that $\vbs(x)\leq v(x)$ for all $x\geq0$. Consequently $\vbs(x)=v(x)$ and by Lemma \ref{hbeta} $v$ must uniquely solve the equation
\eqref{eqvb} with the boundary condition $v^\prime(\beta^*)=1$.
This completes the proof.
\qed

\end{document}